\documentclass[aip,jmp,amsmath,amssymb,preprint,]{revtex4-1}
\usepackage{color}
\usepackage{mathrsfs,amsmath}
\usepackage{mathrsfs}
\usepackage{graphicx}
\usepackage{dcolumn}
\usepackage{bm}
\usepackage{natbib}
\usepackage{amsthm}

\begin{document}
\newtheorem{thm}{Theorem}[section]
\newtheorem{lemma}[thm]{Lemma}
\newtheorem{prop}[thm]{Proposition}
\newtheorem{rem}[thm]{Remark}
\newtheorem{cor}[thm]{Corollary}

\title{The Markov Process Admits a Consistent Steady-State Thermodynamic Formalism}

\author{Liangrong Peng}
\author{Yi Zhu}%
\author{Liu Hong}
\homepage{Author to whom correspondence should be addressed. Electronic mail: zcamhl@tsinghua.edu.cn}
\affiliation{Zhou Pei-Yuan Center for Applied Mathematics, Tsinghua University,
Beijing, China, 100084}

\date{\today}

\begin{abstract}
The seek for a new universal formulation for describing various non-equilibrium processes is a central task of modern non-equilibrium thermodynamics. In this paper, a novel steady-state thermodynamic formalism was established for general Markov processes described by the Chapman-Kolmogorov equation. Furthermore, corresponding formalisms of steady-state thermodynamics for master equation and Fokker-Planck equation could be rigorously derived in mathematics.  To be concrete, we proved that: 1) in the limit of continuous time, the steady-state thermodynamic formalism for the Chapman-Kolmogorov equation fully agrees with that for the master equation; 2) a similar one-to-one correspondence could be established rigorously between the master equation and Fokker-Planck equation in the limit of large system size; 3) when a Markov process is restrained to one-step jump, the steady-state thermodynamic formalism for the Fokker-Planck equation with discrete state variables also goes to that for master equations, as the discretization step gets smaller and smaller. Our analysis indicated that, with respect to the steady state, general Markov processes admit a unified and self-consistent non-equilibrium thermodynamic formulation, regardless of underlying detailed models.
\end{abstract}

\maketitle

\section{Introduction}
How to extend the concepts and methodology of equilibrium thermodynamics to general non-equilibrium cases is a big question. \cite{de2013,jou1999} After more than 80 years of hard working since Onsager published his celebrated reciprocal relation in 1931, \cite{onsager1931_1,onsager1931_2} which is generally recognized as the beginning of modern non-equilibrium thermodynamics, there is no widely accepted unified theory for describing various non-equilibrium processes. An important forward step was done by Oono and Paniconi in 1998. \cite{oono1998} By introducing several new concepts as ``excess heat'' and ``housekeeping dissipation'' for characterizing the energy exchange between a given system and its surroundings, they proposed an phenomenological framework to extend the equilibrium thermodynamics to the non-equilibrium steady state (NESS). The NESS is a direct generalization of thermodynamic equilibrium state, both of which are time irrelevant. However, the NESS allows mass and energy transfer within a system or between systems, just provided the system is still in a dynamical balance. The NESS is usually correlated with key words like open system, heat and mass exchange, circular flows, break down of detailed balance, net entropy production rate etc. \cite{qian2006}

Employing the basic framework of Oono and Paniconi, Sasa and Tasaki \cite{sasa2006} subsequently attempted to search for a universal thermodynamic formalism in non-equilibrium physics, which is expected to apply to a large class of non-equilibrium steady states including a heat conducting fluid, a sheared fluid, and an electrically conducting fluid. Later, the formalism of steady-state thermodynamics has been developed into great details in several concrete classical examples. As an example, for the Langevin dynamics, Sekimoto, \cite{mizutani1997,sekimoto1998} Hatano and Sasa \cite{hatano2001} found that the extended form of the second law holds for transitions between steady states and the Shannon entropy difference is related to the excess heat produced in an infinitely slow operation. As to the master equation and Fokker-Planck equation, Esposito and Van den Broeck \cite{esposito2010,van2010} showed that with respect to the steady state, the entropy production rate could be decomposed into a sum of two non-negative terms, namely the adiabatic and non-adiabatic parts, which reflect the irreversibility of the system under equilibrium and steady states respectively. Based on their formulation, the second law of thermodynamics could be casted into three different strengthened versions. Alternative strengthened versions of second law of thermodynamics have been reported by Hong \textit{et al.} for the master equation too, \cite{hong2016} and then been proved valid for general irreversible processes as an inference of famous KL divergence. \cite{huang2016} Recently, by using the large-deviation theory, Ge and Qian \cite{ge2016} proved that a nonlinear chemical reaction system has a consistent steady state thermodynamic formulism in both macroscopic and mesoscopic scales, including the free energy, entropy production rate and its decomposition.

Although it appears that many classical physical systems could be casted into a unified framework of steady-state thermodynamics proposed by Oono and Paniconi, \cite{oono1998} it is still questionable how closely are those formalisms related to each other. Especially given the deep mathematical and physical correlations among several well-known models, \textit{e.g.}, according to It$\hat{o}$ calculus the Fokker-Planck equation governs the probability distribution evolution of a corresponding Langevin dynamics in time, it would be natural to expect the constructed steady-state thermodynamic formalisms on those models will preserve such a kind of correspondence. Only in this way, we could expect the existence, uniqueness, universality and operability of a unified steady-state thermodynamics for general non-equilibrium processes, rather than studies case by case.

Motivated by recent developments in the steady-state thermodynamics, particularly the works of Esposito and Van den Broeck on the master equation and Fokker-Planck equation, \cite{esposito2010,van2010} we are trying to show there is indeed a universal framework of steady-state thermodynamic descriptions, at least for various Markov processes in the discrete or continuous space. To be concrete, the steady-state thermodynamic formalism here we mean not only includes those classical thermodynamic elements like the internal energy, Helmholtz free energy, Boltzmann entropy, entropy flow and entropy production rate, but also contains the excess heat, decomposition of entropy production
rate into adiabatic and non-adiabatic parts as well as three different strengthened versions of the second law of thermodynamics as defined in Eqs. \eqref{ckentropy}-\eqref{face3} for the Chapman-Kolmogorov equation.

More importantly, with respect to these quantities and relations for steady-state thermodynamics, a one-to-one correspondence during the coarse graining procedure from the Chapman-Kolmogorov equation to the master equation and then to the Fokker-Planck equation could be rigorously established in mathematics. We show that: 1) in the limit of continuous time, the steady-state thermodynamic
formalism for the Chapman-Kolmogorov equation fully agrees with that for
the master equation; 2) a similar one-to-one correspondence could be established rigorously
between the master equation and Fokker-Planck equation in the limit of large
system size; 3) when a Markov process is restrained to one-step jump, the steady-state
thermodynamic formalism for the Fokker-Planck equation with discrete state
variables also goes to the formulation for master equations, as the discretization step gets smaller and smaller. By studying these concrete models of Markov processes, the universality of steady-state thermodynamics could thus be verified.

\section{Steady-state thermodynamic formalism for C-K equations}

As the most important stochastic process, the Markov process has been extensively applied to laser physics, chemical reactions, molecular biology and many other fields. \cite{reichl1980}
The Markov property indicates that, the future probabilities of a Markov process could be solely determined by the present state, hence independent of its whole history. Mathematically, a Markov process is characterized by the well-known Chapman-Kolmogorov equation (or C-K equation for short),
\begin{equation}\label{ckequation}
P(x,t+\Delta t|x_0,t_0)=\int dy P(x,t+\Delta t|y,t)P(y,t|x_0,t_0),\quad  t_0<t<t+\Delta t,
\end{equation}
where $P(x,t+\Delta t|x_0,t_0)$ is the transition probability density at $t+\Delta t$ given the initial position $x_0$ at time $t_0$.
In what follows, the dependence of variables on initial conditions will be dropped for notational simplicity, $i.e.$, $P(x,t)\equiv P(x,t|x_0,t_0)$. Note the integral above can be directly replaced by a summation in the discrete case, without affecting all following results.

After sufficiently long time, a Markov process would be expected to reach a steady state characterized by a time-independent probability density function
\begin{equation*}
P^s(x,t)=P^s(x,t')=P^s(x), \quad \forall t, t'>t_0, \forall x.
\end{equation*}
Substituting it into Eq. \eqref{ckequation}, one has
\begin{equation}\label{cksteadystate}
\int dy P(x,t+\Delta t|y,t)P^s(y)=P^s(x,t+\Delta t)=P^s(x)=\int dy P(y,t+\Delta t|x,t)P^s(x),
\end{equation}
for every $x$ and every $t, \Delta t >0$. In the last equality, the normalization property $\int dy P(y,t+\Delta t|x,t)=1$ is used.
If there is a $P^e(x)>0$ further satisfies
\begin{equation}\label{ckdetailbalance}
P(x,t+\Delta t|y,t)P^e(y)= P(y,t+\Delta t|x,t)P^e(x),\quad  \forall~t, \Delta t>0, \forall~ x,y,
\end{equation}
Eq. \eqref{ckequation} will be called under the condition of detailed balance, which apparently is a special case of the steady state condition.

To construct the steady-state thermodynamics for the C-K equation, we define the entropy function $S^{C}(t)$ as
\begin{equation}\label{ckentropy}
S^{C}(t)=-\int dx P(x,t) \ln P(x,t),
\end{equation}
where the Boltzmann constant $k_B$ is set to be 1. With the entropy function in hand, it is straightforward to calculate the entropy difference between two successive times,
\begin{align*}
{S^{C}(t+\Delta t)-S^{C}(t)}
&=-\int dx [P(x,t+\Delta t) \ln P(x,t+\Delta t) - P(x,t) \ln P(x,t)]\\
&\equiv \Delta S^C(t)+ \Delta I^C(t),
\end{align*}
where $\Delta S^C(t)=- \int dx [P(x,t+\Delta t)-P(x,t)]\ln P(x,t)$ denotes the usual entropy change, and $\Delta I^C(t)= -\int dx P(x,t+\Delta t) \ln \frac{P(x,t+\Delta t)}{P(x,t)}\leq 0$ denotes the information gain per $\Delta t$ by virtue of the prior probability $P(x,t)$.

The following proposition gives an explicit expression of the entropy change for the Markov process described by the C-K equation.

\begin{prop}
For the C-K equation \eqref{ckequation}, the entropy change per $\Delta t$ is given by
\begin{equation}\label{ckentropychange}
{\Delta S^{C}}(t)=\frac{1}{2} \iint dxdy J^C(x,y,t,\Delta t) \ln \frac{P(y,t)}{P(x,t)},
\end{equation}
where $J^C(x,y,t,\Delta t)=P(x,t+\Delta t|y,t)P(y,t) -P(y,t+\Delta t|x,t)P(x,t)$ is the thermodynamic flux of the C-K equation.
\end{prop}

\begin{proof}
By substituting the C-K equation \eqref{ckequation} into the entropy change $\Delta S^C(t)$, one obtains
\begin{align*}
{\Delta S^{C}}(t)
&=- \int dx [P(x,t+\Delta t)-P(x,t)]\ln P(x,t)\\
&=- \iint dxdy [P(x,t+\Delta t|y,t)-\delta(y-x)]P(y,t)\ln P(x,t)\\
&=- \iint dxdy [P(x,t+\Delta t|y,t)-\delta(y-x)]P(y,t)\ln \frac{P(x,t)}{P(y,t)}\\
&=  \frac{1}{2}\iint dxdy [P(x,t+\Delta t|y,t)P(y,t) -P(y,t+\Delta t|x,t)P(x,t)]\ln \frac{P(y,t)}{P(x,t)},
\end{align*}
where the symmetry of the first term with respect to dummy variables $x$ and $y$ is used in the last equality.

This gives the desired result.
\end{proof}

Based on the general framework for the steady-state thermodynamics proposed by Oono and Paniconi,\cite{oono1998} the entropy change ${\Delta S^{C}}(t)$ could be separated into two parts: the entropy production ${\Delta S^{C}_i}(t)$ and entropy flow ${\Delta S^{C}_e}(t)$. The entropy production is always non-negative as a manifestation of the second law of thermodynamics. More interestingly, it can be further decomposed into a sum of two non-negative parts: the adiabatic and non-adiabatic entropy productions according to their different origins, which are stated through the following proposition.
\begin{prop} \label{prop1}
The entropy production ${\Delta S^{C}_i}(t)$ and entropy flow ${\Delta S^{C}_e}(t)$ per $\Delta t$ for the C-K equation are given as
\begin{align}
{\Delta S^{C}_i}(t)
&=\frac{1}{2} \iint dxdy J^C(x,y,t,\Delta t)\ln \frac{P(x,t+\Delta t|y,t)P(y,t)}{P(y,t+\Delta t|x,t)P(x,t)}\geq 0,\\
{\Delta S^{C}_e}(t)
&=\frac{1}{2} \iint dxdy J^C(x,y,t,\Delta t)\ln \frac{P(y,t+\Delta t|x,t)}{P(x,t+\Delta t|y,t)}.
\end{align}
Furthermore, the entropy production ${\Delta S^{C}_i}(t)$ can be decomposed into an adiabatic and a non-adiabatic parts as
\begin{align}
{\Delta S^{C}_{ad}}(t)
&=\frac{1}{2} \iint dxdy J^C(x,y,t,\Delta t)\ln \frac{P(x,t+\Delta t|y,t)P^s(y)}{P(y,t+\Delta t|x,t)P^s(x)}\geq0,\\
{\Delta S^{C}_{na}}(t)
&=\frac{1}{2} \iint dxdy J^C(x,y,t,\Delta t)\ln \frac{P^s(x)P(y,t)}{P^s(y)P(x,t)}\geq 0.
\end{align}
Here $P^s(x)$ denotes the probability density in the steady state.
\end{prop}
\begin{proof}
It is sufficient to prove the non-negativity of ${\Delta S^{C}_{ad}}(t)$ for $t, \Delta t>0$, and ${\Delta S^{C}_{na}}(t)\geq 0$ could be obtained in a similar way.
By rewriting the adiabatic entropy production in a compact form and using the inequality
$-\ln \xi \geq 1-\xi$ for $\xi>0$,
we have
\begin{align*}
{\Delta S^{C}_{ad}}(t)
&=\iint dxdy {P(x,t+\Delta t|y,t)P(y,t)} [-\ln \frac{P(y,t+\Delta t|x,t)P^s(x)}{P(x,t+\Delta t|y,t)P^s(y)}]\\
&\geq \iint dxdy {P(x,t+\Delta t|y,t)P(y,t)} [1 - \frac{P(y,t+\Delta t|x,t)P^s(x)}{P(x,t+\Delta t|y,t)P^s(y)}]\\
&=\iint dxdy {P(x,t+\Delta t|y,t)P(y,t)} - \int dy \frac{P(y,t)}{P^s(y)}\int dx P(y,t+\Delta t|x,t)P^s(x)\\
&=\iint dxdy {P(x,t+\Delta t|y,t)P(y,t)} - \int dy \frac{P(y,t)}{P^s(y)}\int dx P(x,t+\Delta t|y,t)P^s(y)
= 0,
\end{align*}
where the steady state condition \eqref{cksteadystate} is used in the last equation.
\end{proof}

\begin{rem}
The non-negativity of the entropy production ${\Delta S^{C}_i}(t)$ is guaranteed by the second law of thermodynamics.
It becomes zero if and only if the condition of detailed balance Eq. \eqref{ckdetailbalance} holds. ${\Delta S^{C}_e}(t)$ represents the entropy exchange between the system and its surrounding environment and does not have a definite sign.
The adiabatic entropy production, also known as the housekeeping heat, becomes zero if and only if the condition of detailed balance holds; while the non-adiabatic part vanishes as long as the steady state is attained.
\end{rem}

Further introduce the excess entropy change as
\begin{equation}
{\Delta S^{C}_{ex}}(t)
=\frac{1}{2} \iint dxdy J^C(x,y,t,\Delta t)\ln \frac{P^s(y)}{P^s(x)},
\end{equation}
then following results are derived as a direct corollary of above propositions.
\begin{cor}
For any $t, \Delta t>0$, the C-K equation has following relations,
\begin{eqnarray}
&{\Delta S^{C}}(t)-{\Delta S^{C}_e}(t)={\Delta S^{C}_i}(t)\geq0, \label{face1}\\
&{\Delta S^{C}}(t)-{\Delta S^{C}_{ex}}(t)={\Delta S^{C}_{na}}(t)\geq0, \label{face2}\\
&{\Delta S^{C}_{ex}}(t)-{\Delta S^{C}_{e}}(t)={\Delta S^{C}_{ad}}(t)\geq0 \label{face3}.
\end{eqnarray}
\end{cor}

The above corollary presents three different versions of the second law of thermodynamics for general Markov processes. Especially, the later two go beyond the classical one, and indicate that the second law of thermodynamics could be strengthened for certain systems.

Till now we have completed the construction of steady-state thermodynamic formalism for the C-K equation. As we will show later, above results, especially the strengthened versions of second law of thermodynamics, exactly correspond to those for the master equation \cite{esposito2010} and the Fokker-Planck equation \cite{van2010} in the thermodynamic limit, which constitutes the major conclusion of our current paper. Therefore, the steady-state thermodynamics is a self-consistent theory at least for various non-equilibrium systems governed by the Markov process and has more fruitful results than classical equilibrium thermodynamics.

\section{From C-K equation to master equation}

The Chapman-Kolmogorov equation is an manifestation of total probability theorem in discrete time space. If the time interval becomes smaller and smaller, the C-K equation will go to the master equation in the limit of continuous time.
To see this, we restrict ourselves to the stationary (or time-homogenous) Markov process, that is, the conditional probability density $P(x,t+\Delta t|y,t)$ is assumed to be independent of $t$. When $\Delta t \rightarrow 0$, using the Taylor series expansion, one has
\begin{equation}\label{transitionprobabilitytaylor}
P(x,t+\Delta t|y,t)=
\begin{cases}
\Delta t W(x|y)+O((\Delta t)^2)            & x\neq y\\
1-\Delta t \int_{x\neq y}dx W(x|y)+O((\Delta t)^2)  & x=y
\end{cases},
\end{equation}
where $W(x|y)\geq 0$ denotes the transition probability per unit time from state $y$ to state $x$.

Substituting formula \eqref{transitionprobabilitytaylor} to the C-K equation (Eq. \eqref{ckequation}), we arrive at
\begin{equation*}
  \frac{P(x, t+\Delta t)-P(x,t)}{\Delta t}=\int dy [W(x|y)P(y,t)-W(y|x)P(x,t)]+\frac{O((\Delta t)^2) }{\Delta t},\quad \forall t,\Delta t>0.
\end{equation*}
Taking the limit $\Delta t \rightarrow 0$ leads to the master equation, which describes how the probability of a system in state $x$ at time $t$ evolves with time, \textit{i.e.},
\begin{equation}\label{masterequation}
  \frac{\partial}{\partial t}P(x,t)=\int dy [W(x|y)P(y,t)-W(y|x)P(x,t)].
\end{equation}
The master equation is an ordinary differential equation when the state space is discrete and becomes an integral-differential equation in the case of continuous states.
It has been widely used in various stochastic processes, including random walks, birth-death processes, general chemical reaction systems, \cite{kurtz1972,gillespie1992} thermal unimolecular reactions at low pressures, \cite{troe1977} single-molecule enzyme kinetics in open biochemical systems \cite{ge2012} \textit{etc}.

By utilizing the formula \eqref{transitionprobabilitytaylor}, the steady state of the C-K equation in Eq. \eqref{cksteadystate} for all $x \neq y$ becomes
\begin{equation*}
\lim_{\Delta t \rightarrow 0} \frac{1}{\Delta t} \int dy \big[\Delta t W(x|y)P^s(y)- \Delta t W(y|x)P^s(x)+ O((\Delta t)^2) \big]=0,
\end{equation*}
thus,
\begin{equation}\label{mesteadystate}
\int dy [W(x|y)P^s(y)-W(y|x)P^s(x)]=0,\quad \forall x,
\end{equation}
since the above equation holds automatically when $x=y$. Clearly, Eq. \eqref{mesteadystate} agrees with the usual definition of steady-state solution $P^s(x)$ for the master equation. Similarly, the condition of detailed balance of the C-K equation becomes
\begin{equation}
W(x|y)P^e(y)=W(y|x)P^e(x),\quad \forall x,y,
\end{equation}
which is consistent with that for the master equation too.

Now we want to explore whether the steady-state thermodynamic formalism defined for the C-K equation (Eqs. \eqref{ckentropy}-\eqref{face3}) is also valid for the master equation. As the entropy function used for the master equation
\begin{equation}
S^{M}(t)=-\int dx P(x,t) \ln P(x,t)
\end{equation}
remains the same as that for the C-K equation, we can directly make a use of Eq. \eqref{transitionprobabilitytaylor} and obtain the following lemma.
\begin{lemma}\label{lemma1}
In the limit of continuous time, with respect to the transition rates defined in \eqref{transitionprobabilitytaylor}, the information gain for the stationary Markov process described by the C-K equation in \eqref{ckequation} vanishes, and the entropy change rate approaches to the entropy difference per unit time, $i.e.,$
\begin{equation*}
\lim_{\Delta t \rightarrow 0}\frac{\Delta I^{C}(t)}{\Delta t}=0,
 \quad
\lim_{\Delta t \rightarrow 0} \frac{\Delta S^C(t)}{\Delta t}=\lim_{\Delta t \rightarrow 0}\frac{S^{C}(t+\Delta t)-S^C(t)}{\Delta t}.
\end{equation*}
\end{lemma}
\begin{proof}
Notice that
\begin{equation*}
{P(x, t+\Delta t)-P(x,t)}={\Delta t} \int dy [W(x|y)P(y,t)-W(y|x)P(x,t)]+{O((\Delta t)^2) },
\end{equation*}
then
\begin{align*}
\lim_{\Delta t \rightarrow 0} \frac{\Delta I^{C}(t)}{\Delta t}
&=\lim_{\Delta t \rightarrow 0} \frac{1}{\Delta t}\int dx P(x,t+\Delta t)\ln \bigg[1- \frac{{P(x, t+\Delta t)-P(x,t)}}{P(x, t+\Delta t)}\bigg]\\
&=-\lim_{\Delta t \rightarrow 0}\frac{1}{\Delta t}\int dx [{P(x, t+\Delta t)-P(x,t)}] =0.
\end{align*}
This shows that the entropy change rate equals to the entropy difference per unit time in the limit of $\Delta t$.
\end{proof}

Based on Proposition \ref{prop1} and Lemma \ref{lemma1}, we have following results.
\begin{thm}\label{thm1}
In the limit of continuous time, with respect to the transition rates defined in \eqref{transitionprobabilitytaylor}, the  entropy production rate, entropy flow rate, adiabatic entropy production rate, non-adiabatic entropy production rate and excess entropy change rate for the C-K equation \eqref{ckequation} become
\begin{align}
\lim_{\Delta t \rightarrow 0}\frac{\Delta S^{C}_i }{\Delta t}&=\frac{1}{2} \iint dxdy J^M(x,y,t)\ln \frac{W(x|y)P(y,t)}{W(y|x)P(x,t)}\equiv \frac{dS^{M}_i }{dt}\geq0,\\
\lim_{\Delta t \rightarrow 0}\frac{\Delta S^{C}_e }{\Delta t}&=\frac{1}{2} \iint dxdy J^M(x,y,t)\ln \frac{W(y|x)}{W(x|y)}\equiv  \frac{dS^{M}_e }{dt},\\
\lim_{\Delta t \rightarrow 0} \frac{\Delta S^{C}_{ad}}{\Delta t}&=\frac{1}{2} \iint dxdy J^M(x,y,t)\ln \frac{W(x|y)P^s(y)}{W(y|x)P^s(x)}\equiv  \frac{d S^{M}_{ad}}{dt}\geq0,\\
\lim_{\Delta t \rightarrow 0} \frac{\Delta S^{C}_{na}}{\Delta t}&=\frac{1}{2} \iint dxdy J^M(x,y,t)\ln \frac{P^s(x)P(y,t)}{P^s(y)P(x,t)}\equiv  \frac{d S^{M}_{na}}{dt}\geq0,\\
\lim_{\Delta t \rightarrow 0} \frac{\Delta S^{C}_{ex}}{\Delta t}&=\frac{1}{2} \iint dxdy J^M(x,y,t)\ln \frac{P^s(y)}{P^s(x)}\equiv  \frac{dS^{M}_{ex}}{dt},
\end{align}
which emerge as the entropy production rate, entropy flow rate, adiabatic entropy production rate, non-adiabatic entropy production rate and excess entropy change rate for the master equation \eqref{masterequation}, respectively. Here, the thermodynamic flux is defined as $J^M(x,y,t)=W(x|y)P(y,t) - W(y|x)P(x,t)$.
\end{thm}
Note that in the limit of continuous time, the steady-state thermodynamic formalism defined for the C-K equation fully agrees with the formulation for the discrete master equation.\cite{esposito2010,van2010} Thus the correspondence between the C-K equation and the master equation on the steady-state thermodynamics is completely verified.

\section{From master equation to F-P equation}
The master equation involves transitions among all possible states, which are hard to be modeled or computed.
Practically, a coarser description of the system in replace of the master equation is needed, which is now known as the Fokker-Planck equation (F-P equation for short). The F-P equation has been applied to fields as diverse as quantum optics,\cite{risken1989} micro-macro coupling models of polymeric fluids,\cite{keunings2004} biochemical oscillations, \cite{cao2015}  electric circuits and laser arrays, population dynamics and stock marketing \cite{frank2005} \textit{etc.}.

It is well known that the F-P equation can be deduced from the master equation by expanding the transition rates and neglecting high order terms of jump moments.
In order to make the derivation strictly in mathematics, we adopt the canonical form expansion introduced by Van Kampen, \cite{van1983} in which a parameter $\Omega$ representing the system size is introduced. Then the master equation \eqref{masterequation} can be rewritten as
\begin{equation}\label{meomega}
  \frac{\partial}{\partial t}P_{\Omega}(x,t)=\int_{V_x} dy [W_{\Omega}(x|y)P_{\Omega}(y,t)-W_{\Omega}(y|x)P_{\Omega}(x,t)].
\end{equation}
Without loss of generality, we set $V_x=\{y:~ \parallel y-x \parallel \leq \Omega \}$ as a region with centre $x$ and radius $\Omega$, the dependence of probability $P_{\Omega}$ and transition rate $W_{\Omega}$ on ${\Omega}$ is written out explicitly. By defining a jump length $r=x-y$, one can rewrite $W_{\Omega}(x|y) \equiv W_{\Omega}(y,r)=W_{\Omega}(x-r,r)$.

With respect to the system size $\Omega$, we introduce the re-scaled state variable, time and jump distance respectively as
\begin{equation}\label{scaledvariable}
X={\Omega}^{-1}x,\quad T={\Omega}^{-1}t,\quad R={\Omega}^{-1}r.
\end{equation}
Consequently, we have $dx=\Omega dX, {\partial}_t={\Omega}^{-1}{\partial}_T, dy=-\Omega dR$.
Since $W_{\Omega}(x-r,r)$ and  $P_{\Omega}(x,t)$ can be expanded into power series of ${\Omega}^{-1}$, \textit{i.e.}
\begin{eqnarray}
&&W_{\Omega}(x-r,r)= W_{\Omega}(\Omega(X-R),\Omega R)=\Phi_0(X-R,R) + {\Omega}^{-1}\Phi_1(X-R,R) + {\Omega}^{-2}\Phi_2(X-R,R) + \cdots, \nonumber\\
&&P_{\Omega}(x,t)= P_{\Omega}(\Omega X,{\Omega}T)=p_0(X,T) + {\Omega}^{-1}p_1(X,T) + {\Omega}^{-2}p_2(X,T)+\cdots, \label{pomega}
\end{eqnarray}
the master equation \eqref{meomega} is reformulated as
\begin{align*}
&\frac{\partial}{\partial T}\big[p_0(X,T) + {\Omega}^{-1}p_1(X,T) + {\Omega}^{-2}p_2(X,T)+\cdots \big] \\
&={\Omega}^{2} {\int_{V_0}} dR  \big[\Phi_0(X-R,R) + {\Omega}^{-1}\Phi_1(X-R,R) + {\Omega}^{-2}\Phi_2(X-R,R) + \cdots\big] \\
& \quad\quad \quad\times \big[ p_0(X-R,T) + {\Omega}^{-1}p_1(X-R,T) + {\Omega}^{-2}p_2(X-R,T)+\cdots \big ] \\
&-{\Omega}^{2} {\int_{V_0}} dR  \big[\Phi_0(X,-R) + {\Omega}^{-1}\Phi_1(X,-R) + {\Omega}^{-2}\Phi_2(X,-R) + \cdots\big] \\
&\quad \quad\quad\times \big[ p_0(X,T) + {\Omega}^{-1}p_1(X,T) + {\Omega}^{-2}p_2(X,T)+\cdots \big].
\end{align*}
Here $V_0=\{R:~ \parallel R \parallel \leq 1 \}$ denotes a region with centre $0$ and radius $1$. Taking Taylor series expansion of $\Phi_i(X-R,R)$ and $p_i(X-R,T)$ 
with respect to $X$, we have
\begin{align}
&\frac{\partial}{\partial T}\big[p_0(X,T) + {\Omega}^{-1}p_1(X,T) + {\Omega}^{-2}p_2(X,T)+\cdots \big]  \nonumber \\
&=\sum_{i,j=0}^\infty {\Omega}^{2-i-j} {\int_{V_0}} dR  \big\{
-R \frac{\partial}{\partial X} [\Phi_i(X,R)p_j(X,T)] + \frac{R^2}{2} \frac{{\partial}^2}{\partial X^2} [\Phi_i(X,R)p_j(X,T)]+\cdots \big\}, \label{equationtaylor}
\end{align}
where terms $\int_{V_0} dR [\Phi_i(X,R)-\Phi_i(X,-R)]$ $(i\geq0)$ are exactly cancelled due to the symmetry of integral region ${V_0}$.

To make a coarse graining of the system, following jump moments are introduced
\begin{equation}
\alpha_{j,i}(X)=\int_{V_0}dR \frac{{\Omega}^j R^j}{j!} \Phi_{i}(X,R).
\end{equation}
Especially, here we are interested in master equations of the diffusion-type, \cite{van1983} which require $\alpha_{1,0}(X)=0$. By substituting the jump moments into Eq. \eqref{equationtaylor} and taking the limit $\Omega \rightarrow \infty$, the zeroth order equation in $\Omega$ yields the casual Fokker-Planck equation
\begin{equation}\label{fpequation}
\frac{\partial}{\partial T} p_0(X,T)
=-\frac{\partial}{\partial X} J^F(X,T), \quad J^F(X,T)=\alpha_{1,1}(X)p_0(X,T)-\frac{\partial}{\partial X}[\alpha_{2,0}(X)p_0(X,T)],
\end{equation}
where $J^F(X,T)$ is the probability flux. $\alpha_{1,1}(X)$ and $\alpha_{2,0}(X)$ ($\alpha_{2,0}(X)>0$) denote the drift and diffusion coefficients respectively.
Note that both coefficients $\alpha_{1,1}(X)$ and $\alpha_{2,0}(X)$ are functions of position $X$.

Similarly, the steady state of the master equation \eqref{meomega} becomes
\begin{align*}
 \lim_{\Omega\rightarrow\infty}\sum_{i,j=0}^\infty {\Omega}^{2-i-j} \displaystyle{\int_{V_0}} dR  \big\{
-R \frac{\partial}{\partial X} [\Phi_i(X,R)p_j^s(X)] + \frac{R^2}{2} \frac{{\partial}^2}{\partial X^2} [\Phi_i(X,R)p_j^s(X)]+\cdots \big\}=0,
\end{align*}
which gives
\begin{equation}
\frac{\partial}{\partial X} \big\{ \alpha_{1,1}(X)p^s_0(X)-\frac{\partial}{\partial X}[\alpha_{2,0}(X)p^s_0(X)] \big\}=0,\quad \forall X.
\end{equation}
This is exactly the definition of the steady state for the F-P equation. Thus we have justified the consistency of the steady state between the master equation and the F-P equation.
A subclass of steady state is the detailed balance $p_0^e(X)$, which satisfies
\begin{equation}
\alpha_{1,1}(X)p^e_0(X)-\frac{\partial}{\partial X}[\alpha_{2,0}(X)p^e_0(X)]=0, \quad \forall X.
\end{equation}

Starting from the entropy function for the master equation $S^M(t) = -\int dx P_{\Omega}(x,t)\ln P_{\Omega}(x,t)$, it is straightforward to verify that the volume density of entropy $S^M$
\begin{equation}
\lim_{\Omega \rightarrow \infty} {\Omega}^{-1}  S^M(t)=-\int dX p_0(X,T)\ln p_0(X,T)\equiv S^F(T)
\end{equation}
emerges as the entropy function for the F-P equation in \eqref{fpequation}. This is not a coincidence. Actually, we can further show the steady-state thermodynamic formalism defined on the master equation converges automatically to that on the F-P equation in the limit of large system size $\Omega \rightarrow \infty$. This interesting correspondence confirms that there indeed exists a universal thermodynamic framework valid for general non-equilibrium phenomena described by Markov processes with respect to the steady state. Although a similar conclusion has been reached in Esposito and Van den Broeck's original paper, \cite{van2010} their results are limited to a Markov chain process with only one-step jump and can not be generalized to cases with multiple-step jump in principle. In contrast, our proof and conclusions do not suffer from such a limitation and are more rigorous in mathematics.

\begin{thm}\label{thm2}
 In the limit of large system size $\Omega \rightarrow \infty$, the volume density of entropy production rate, entropy flow rate, adiabatic entropy production rate, non-adiabatic entropy production rate and excess entropy change rate for the master equation \eqref{meomega} become
\begin{align}
\lim_{\Omega \rightarrow \infty} \bigg({\Omega}^{-1} \frac{dS^M_i}{dT}\bigg)&=\int dX     J^F\cdot \frac{J^F} { \alpha_{2,0}(X)p_0} \equiv \frac{dS^F_i}{dT}\geq 0,\label{eprintheorem1}\\
\lim_{\Omega \rightarrow \infty} \bigg({\Omega}^{-1} \frac{dS^M_e}{dT}\bigg)&=\int dX     J^F\cdot \frac{{\partial_X}\alpha_{2,0}(X)-\alpha_{1,1}(X)}{ \alpha_{2,0}(X)}\equiv \frac{dS^F_e}{dT},\\
\lim_{\Omega \rightarrow \infty} \bigg({\Omega}^{-1} \frac{dS^M_{ad}}{dT}\bigg)&= \int dX J^F\cdot \frac{\alpha_{1,1}(X)p^s_0-{\partial_X}[\alpha_{2,0}(X)p^s_0] }{\alpha_{2,0}(X) p_0^s}\equiv \frac{dS^F_{ad}}{dT}\geq 0,\\
\lim_{\Omega \rightarrow \infty} \bigg({\Omega}^{-1} \frac{dS^M_{na}}{dT}\bigg)&=\int dX  J^F\cdot [{\partial_X} (\ln p_0^s)- {\partial_X} (\ln p_0)] \equiv \frac{dS^F_{na}}{dT} \geq 0,\\
\lim_{\Omega \rightarrow \infty} \bigg({\Omega}^{-1} \frac{dS^M_{ex}}{dT}\bigg)&=\int dX  J^F\cdot [-\partial_X (\ln p_0^s)] \equiv \frac{dS^F_{ex}}{dT},
\end{align}
which emerge as the entropy production rate, entropy flow rate, adiabatic entropy production rate, non-adiabatic entropy production rate and excess entropy change rate for the Fokker-Planck equation \eqref{fpequation}, respectively.
\end{thm}

\begin{proof}
Here we take the instantaneous entropy production rate as an example. Other relations could be deduced in a similar way, please see Appendix for details.
We rewrite the instantaneous entropy production rate for the master equation \eqref{meomega} as
\begin{equation*}
\frac{dS^M_i}{dt}=\frac{1}{2} \iint dxdy J^{M}(x,y,t)A_1^{M}(x,y,t),
\end{equation*}
where $J^{M}(x,y,t)=W_{\Omega}(x|y)P_{\Omega}(y,t) - W_{\Omega}(y|x)P_{\Omega}(x,t)$ and $A_1^{M}(x,y,t)=\ln [W_{\Omega}(x|y)P_{\Omega}(y,t)]-\ln [W_{\Omega}(y|x)P_{\Omega}(x,t)]$ represent the thermodynamic flux and force for the master equation \eqref{meomega} respectively. Expand $J^{M}$ into Taylor series with respect to $X$ as
\begin{align*}
{J^{M}}(x,y,t)
&=\big[\Phi_0(X-R,R) + {\Omega}^{-1}\Phi_1(X-R,R) \big] \cdot \big[ p_0(X-R,T) + {\Omega}^{-1}p_1(X-R,T) \big] \\
& -\big[ \Phi_0(X,-R) + {\Omega}^{-1}\Phi_1(X,-R) \big] \cdot \big[ p_0(X,T) + {\Omega}^{-1}p_1(X,T) \big] + O({\Omega}^{-2}) \\
&=[\Phi_0(X,R)-\Phi_0(X,-R)]p_0  -{R}{\partial_X} [\Phi_0(X,R)p_0] + {\Omega}^{-1} [\Phi_1(X,R)-\Phi_1(X,-R)]p_0  \\
& +{\Omega}^{-1}[\Phi_0(X,R)-\Phi_0(X,-R)]p_1 + O({\Omega}^{-2})\equiv I_{11}(X,R)+ O({\Omega}^{-2}),
\end{align*}
where $p_i$ is short for $p_i(X,T)$. Similarly,
\begin{equation*}
W_{\Omega}(y|x)P_{\Omega}(x,t)
=\Phi_0(X,-R)p_0 + {\Omega}^{-1}\Phi_1(X,-R)p_0+ {\Omega}^{-1}\Phi_0(X,-R)p_1 + O({\Omega}^{-2})
\equiv I_{12}(X,R)+ O({\Omega}^{-2}).
\end{equation*}
With respect to above formulas, the thermodynamic force $A_1^{M}$ can be expanded as
\begin{equation*}
A_{1}^{M}(x,y,t)
=\ln \bigg[ 1+ \frac{J^{M}(x,y,t)}{W_{\Omega}(y|x)P_{\Omega}(x,t)} \bigg]
=\frac{I_{11}(X,R)}{I_{12}(X,R)}+O({\Omega}^{-2}).
\end{equation*}
Note that $I_{11} \sim O({\Omega}^{-1})$ and $I_{12} \sim O(1)$.
Accordingly, the volume density of entropy production rate for the master equation becomes
\begin{align*}
{\Omega}^{-1} \frac{dS^M_i}{dT}
&=\frac{1}{2}{\Omega}^{2} \iint dXdR \big[I_{11}(X,R)+ O({\Omega}^{-2})\big] \cdot \bigg[\frac{{I_{11}(X,R)}}{I_{12}X,R)}+O({\Omega}^{-2})\bigg]\\
&=\frac{{\Omega}^{2}}{2} \iint dXdR\bigg[\frac{{I_{11}(X,R)}^2}{I_{12}(X,R)}+O({\Omega}^{-3})\bigg].
\end{align*}
Since ${I_{12}(X,R)} \geq 0$, according to the Cauchy-Schwarz inequality,
\begin{equation}\label{eprinequality}
\frac{{\Omega}^{2}}{2} \iint dXdR\frac{{I_{11}(X,R)}^2}{I_{12}(X,R)}
= \frac{1}{2} \iint dXdR \frac{[\Omega I_{11}(X,R)R]^2}{I_{12}(X,R)R^2}
\geq \frac{1}{2} \int dX \frac{\big \{\Omega \int_{V_0} dR [ I_{11}(X,R)R] \big\}^2}{ \int_{V_0} dR[I_{12}(X,R)R^2]},
\end{equation}
with
\begin{align*}
\Omega \int_{V_0} dR~ [ I_{11}(X,R)R]
&=\Omega \int_{V_0} dR~  [\Phi_0(X,R)-\Phi_0(X,-R)]R p_0
-\Omega\int_{V_0} dR~ {\partial_X} [\Phi_0(X,R)p_0]{R^2}\\
&+\int_{V_0} dR~ [\Phi_1(X,R)-\Phi_1(X,-R)]Rp_0 +
\int_{V_0} dR~ [\Phi_0(X,R)-\Phi_0(X,-R)]R p_1\\
&= 2\alpha_{1,0}(X)p_0 - 2{\Omega}^{-1} \partial_X [\alpha_{2,0}(X)p_0]+2{\Omega}^{-1}\alpha_{1,1}(X)p_0+2{\Omega}^{-1}\alpha_{1,0}(X)p_1\\
&= 2{\Omega}^{-1}\alpha_{1,1}(X)p_0-2{\Omega}^{-1} \partial_X [\alpha_{2,0}(X)p_0],
\end{align*}
and
\begin{align*}
\int_{V_0} dR~[I_{12}(X,R)R^2]
&=\int_{V_0} dR~\big[\Phi_0(X,-R)R^2p_0 + {\Omega}^{-1}\Phi_1(X,-R)R^2p_0+ {\Omega}^{-1}\Phi_0(X,-R)R^2p_1\big] \\
&=2{\Omega}^{-2}\alpha_{2,0}(X)p_0+O({\Omega}^{-3}).
\end{align*}
As a result, the lower bound of the volume density of instantaneous entropy production rate for the master equation is given by
\begin{align*}\label{27.1}
{\Omega}^{-1} \frac{dS^M_i}{dT}\geq\frac{1}{2} \int dX  \big[\frac{J^F(X,T)^2} { \alpha_{2,0}(X)p_0(X,T)}+O(\Omega^{-1})\big].
\end{align*}
The dominant term on the right-hand side is exactly the instantaneous entropy production rate for the F-P equation. Actually, above inequality provides an interesting relation between the entropy production rate for the master equation and that for the F-P equation. The latter exists as a lower bound for the former, which agrees with the information loss during the coarse graining procedure.

In the limit of large system size $\Omega\rightarrow\infty$, the transition rates $\Phi_i(X,R)$ approach to delta functions with respect to the jump length $R$ (since $R=r/\Omega\rightarrow0$). Then the equality in \eqref{eprinequality} holds, meaning
\begin{align*}
\lim_{\Omega \rightarrow \infty} \bigg({\Omega}^{-1}   \frac{dS^M_i}{dT}\bigg)&=\int dX  \frac{J^F(X,T)^2} { \alpha_{2,0}(X)p_0(X,T)} \equiv \frac{dS^F_i}{dT}\geq 0.
\end{align*}
This completes our proof.
\end{proof}

\begin{rem}
From Theorem \ref{thm1} and \ref{thm2}, it is notable that the instantaneous entropy-production rate, entropy flow rate, adiabatic entropy production rate, non-adiabatic entropy production rate and excess entropy change rate for Markov processes could all be expressed as a bilinear form of thermodynamic fluxes and forces, a reflection of the famous Onsager-Casimir relation. \cite{onsager1931_1}
\end{rem}

\section{From F-P equation to master equation}
In the last section, we deduce the steady-state thermodynamic structure for the F-P equation from that of the master equation. Astonishingly, an inverse procedure is also valid in mathematics. To be concrete, we can derive the steady-state thermodynamic formalism for a special type of master equations with a tridiagonal transition rate matrix from that of the corresponding F-P equation. This is the best result one can expected, since only the first two jump moments are kept in the F-P equation.

We start with the F-P equation in \eqref{fpequation} and discretize its right-hand side with respect to space variable $X=n\epsilon$.
\begin{align*}
\frac{dp(n,T)}{d T}
&=-\frac{1}{\epsilon} [\alpha_{1,1}(n+1) p(n+1,T)-\alpha_{1,1}(n) p(n,T)] \\
&+ \frac{1}{\epsilon^2}[\alpha_{2,0}(n+1)p(n+1,T)-2\alpha_{2,0}(n)p(n,T)+\alpha_{2,0}(n-1)p(n-1,T)],
\end{align*}
for $p(n,T) \equiv p_0(n\epsilon,T)$, $\alpha_{1,1}(n) \equiv \alpha_{1,1}(n\epsilon)$ and $\alpha_{2,0}(n) \equiv \alpha_{2,0}(n\epsilon)$. Clearly, the small parameter $\epsilon$ corresponds to $1/\Omega$ in the canonical form expansion. The above equation can be rewritten into a master equation
\begin{equation}\label{discreteme}
\frac{dp(n,T)}{d T}
=[W(n|n-1) p(n-1,T)-W(n-1|n) p(n,T)] + [W(n|n+1) p(n+1,T)-W(n+1|n) p(n,T)],
\end{equation}
with the forward and backward transition rates
\begin{equation*}
 W(n|n-1)=\frac{\alpha_{2,0}(n-1)}{\epsilon^2},\quad W(n-1|n)=\frac{\alpha_{2,0}(n)}{\epsilon^2}-\frac{\alpha_{1,1}(n)}{\epsilon}.
\end{equation*}
Then based on the steady-state thermodynamic formalism for the F-P equation, following results can be verified.

\begin{thm}\label{thm3}
In the limit of $\epsilon \rightarrow 0$, the entropy production rate, entropy flow rate, adiabatic entropy production rate, non-adiabatic entropy production rate and excess entropy change rate for the Fokker-Planck equation \eqref{fpequation} with discrete state variables become
\begin{align}
\lim_{\epsilon \rightarrow 0} \frac{d S^{F}_i(\epsilon) }{d T}
&= \frac{1}{2}\sum_{m,n}  J^M(m,n,T)\ln \frac{W(n|m) p(m,T)}{W(m|n) p(n,T)} = \frac{dS^M_i}{dT} \geq 0,\nonumber\\
\lim_{\epsilon \rightarrow 0} \frac{d S^{F}_e(\epsilon) }{d T}
&= \frac{1}{2}\sum_{m,n}  J^M(m,n,T) \ln \frac{W(m|n)}{W(n|m)} = \frac{dS^M_e}{dT} ,\nonumber\\
\lim_{\epsilon \rightarrow 0} \frac{d S^{F}_{ad}(\epsilon)}{d T}
&=\frac{1}{2}\sum_{m,n}  J^M(m,n,T)\ln \frac{W(n|m)P^s(m)}{W(m|n)P^s(n)}=  \frac{d S^{M}_{ad}}{dT}\geq 0,\nonumber\\
\lim_{\epsilon \rightarrow 0} \frac{d S^{F}_{na}(\epsilon)}{d T}
&=\frac{1}{2}\sum_{m,n}  J^M(m,n,T)\ln \frac{P^s(n)P(m,t)}{P^s(m)P(n,t)}=  \frac{d S^{M}_{na}}{dT}\geq 0,\nonumber\\
\lim_{\epsilon \rightarrow 0} \frac{d S^{F}_{ex}(\epsilon)}{d T}
&=\frac{1}{2}\sum_{m,n}  J^M(m,n,T)\ln \frac{P^s(m)}{P^s(n)}=  \frac{dS^{M}_{ex}}{dT},\nonumber
\end{align}
which lead to corresponding results for the master equation \eqref{discreteme}, respectively. Here $J^M(m,n,T)=W(n|m)p(m,T)-W(m|n)p(n,T)$ and $m=n \pm 1$.
\end{thm}

\begin{proof}
We take the entropy production rate $ {dS^F_i}(\epsilon)/{dT}$ as an example. According to Eq. \eqref{eprintheorem1}, the entropy production rate for the F-P equation is given by
\begin{align*}
\frac{dS^F_i}{dT}
= \int {dX} \frac{\big \{ \partial_X[\alpha_{2,0}(X) p_0(X,T)]- \alpha_{1,1}(X)p_0(X,T)\big \}^2}{\alpha_{2,0}(X)p_0(X,T)}.
\end{align*}
Discretize above formula with respect to state variable $X=n \epsilon$ and denote it as
\begin{align*}
\frac{dS^F_i(\epsilon)}{dT}
=\sum_{n}
\frac{\big \{[\alpha_{2,0}(n+1)p(n+1,T)-\alpha_{2,0}(n)p(n,T)]-\epsilon \alpha_{1,1}(n)p(n,T)\big \}^2}{\epsilon^2 \alpha_{2,0}(n)p(n,T)}
\equiv \sum_{n}  D_1(n,T)D_2(n,T),
\end{align*}
where
\begin{align*}
D_1(n,T)
&\equiv\frac{1}{\epsilon}  {\big \{[\alpha_{2,0}(n+1)p(n+1,T)-\alpha_{2,0}(n)p(n,T)]-\epsilon \alpha_{1,1}(n)p(n,T)\big \}}  \\
&=\frac{1}{\epsilon}  {\big \{ [\alpha_{2,0}(n+1)p(n+1,T)-\alpha_{2,0}(n)p(n,T)]-\epsilon \alpha_{1,1}(n+1)p(n+1,T)+ O({\epsilon^2})\big \}}  \\
&={\epsilon} [W(n|n+1)p(n+1,T)-W(n+1|n)p(n,T)]+O(\epsilon),\\
D_2(n,T)
&\equiv\frac{[\alpha_{2,0}(n+1)p(n+1,T)-\alpha_{2,0}(n)p(n,T)]-\epsilon \alpha_{1,1}(n)p(n,T)}{\epsilon \alpha_{2,0}(n)p(n,T)}\\
&=\frac{1}{\epsilon} \frac{[\alpha_{2,0}(n+1)p(n+1,T)-\alpha_{2,0}(n)p(n,T)]-\epsilon \alpha_{1,1}(n+1)p(n+1,T) + O({\epsilon^2})}{\alpha_{2,0}(n)p(n,T)}\\
&=\frac{1}{\epsilon}   \ln \frac{\alpha_{2,0}(n+1) p(n+1,T)-\epsilon \alpha_{1,1}(n+1)p(n+1,T)}{\alpha_{2,0}(n)p(n,T)}+ O({\epsilon}) \\
&=\frac{1}{\epsilon}  \ln \frac{W(n|n+1) p(n+1,T)}{W(n+1|n) p(n,T)}+ O({\epsilon}).
\end{align*}
Note that the leading terms of $D_1(n,T)$ and $D_2(n,T)$ are both $O(1)$.
As a result,
\begin{align*}
\lim_{\epsilon\rightarrow 0} \frac{dS^F_i(\epsilon)}{dT}
&= \lim_{\epsilon\rightarrow 0}
   \sum_{n}  \big[{\epsilon} W(n|n+1)p(n+1,T)-{\epsilon}W(n+1|n)p(n,T)+ O({\epsilon}) \big]\\
&            \times \big[\frac{1}{\epsilon} \ln \frac{W(n|n+1) p(n+1,T)}{W(n+1|n) p(n,T)}+ O({\epsilon}) \big] \\
&= \sum_{n}  \big[ W(n|n+1)p(n+1,T)-W(n+1|n)p(n,T)\big]
                 \ln \frac{W(n|n+1) p(n+1,T)}{W(n+1|n) p(n,T)}  \\
&= \frac{1}{2}\sum_{n}  \big[ W(n|n+1)p(n+1,T)-W(n+1|n)p(n,T)\big]
                 \ln \frac{W(n|n+1) p(n+1,T)}{W(n+1|n) p(n,T)}  \\
&+\frac{1}{2}\sum_{n}  \big[ W(n|n-1)p(n-1,T)-W(n-1|n)p(n,T)\big]
                 \ln \frac{W(n|n-1) p(n-1,T)}{W(n-1|n) p(n,T)}.
\end{align*}
This completes the proof. Other relations could be verified in the same way and will not be shown here.
\end{proof}

\begin{rem}
Based on Theorem \ref{thm2} and \ref{thm3}, we can conclude that, in the limit of large system size, when a Markov process is restricted to one-step jump, the steady-state thermodynamic formalisms for the master equation and for the F-P equation have a one-to-one correspondence.
\end{rem}

\section{Internal energy and Helmholtz free energy}
The entropy and its time evolution are discussed in above sections. It remains to explore the internal energy and Helmholtz free energy.
Firstly, based on the definition of steady state \eqref{cksteadystate}, the internal energy $U^C(t)$ and free energy $F^C(t)$ for the C-K equation \eqref{ckequation} are introduced as
\begin{align}
U^C(t)=&-T \int dx P(x,t)\ln P^s(x),\label{ckinternalenergy}\\
F^C(t)=&U^C(t)-TS^C(t)= T \int dx P(x,t)\ln \frac{P(x,t)}{P^s(x)}\label{ckfreeenergy},
\end{align}
here $T$ denotes the temperature and is assumed to be constant. Accordingly, we have
\begin{align*}
&\Delta U^C
=- \frac{T}{2}  \iint dxdy J^C(x,y,t,\Delta t)\ln \frac{P^s(x)}{P^s(y)} =   T \Delta S_{ex}^C,  \\
& \Delta F^C
=- \frac{T}{2}  \iint dxdy J^C(x,y,t,\Delta t)\ln \frac{P^s(x)P(y,t)}{P^s(y)P(x,t)}= -T \Delta S_{na}^C \leq 0.
\end{align*}	
Therefore, the changes of internal energy and free energy are directly proportional to the excess entropy change and non-adiabatic entropy production, except for a constant factor $T$. Above relations also holds for the master equation and F-P equation, as long as the corresponding internal energy and free energy function are defined in the same form as in Eqs. \eqref{ckinternalenergy} and \eqref{ckfreeenergy}. Consequently, following relations could be established.

\begin{prop} Following relations hold
\begin{eqnarray}
&&\lim_{\Delta t \rightarrow 0} \frac{\Delta U^C}{\Delta t}= \frac{dU^M}{d t}, \qquad\qquad
\lim_{\Delta t \rightarrow 0} \frac{\Delta F^C}{\Delta t}= \frac{dF^M}{d t}\leq 0,  \\
&&\lim_{\Omega \rightarrow \infty} \bigg({\Omega}^{-1} \frac{dU^M}{dT} \bigg)=\frac{dU^F}{dT}, \quad
\lim_{\Omega \rightarrow \infty} \bigg({\Omega}^{-1} \frac{dF^M}{dT} \bigg)=\frac{dF^F}{dT} \leq 0,
\end{eqnarray}
where $\frac{dU^M}{dt}$, $\frac{dF^M}{dt}$, $\frac{dU^F}{dT}$ and $\frac{dF^F}{dT}$ represent the internal energy change rate and free energy dissipation rate for the master equation and F-P equation respectively.
\end{prop}

\section{conclusions and discussions}
In this paper, we have explored the steady-state thermodynamics for Markov processes described by the Chapman-Kolmogorov equation, master equation and Fokker-Planck equation separately.
By taking the limit of continuous time, the steady-state thermodynamic formalism specified in Eqs. \eqref{ckentropy}-\eqref{face3} for the C-K equation fully agrees with that for the master equation.
A similar one-to-one correspondence could be established rigorously between the master equation and F-P equation in the limit of large system size by the canonical form expansion. Furthermore, when a Markov process is restrained to one-step jump, the steady-state thermodynamic formalism for the F-P equation with discrete state variables also goes to the formulation for a special type of master equations with a tridiagonal transition rate matrix, as the discretization step gets smaller and smaller. These interesting connections show that the steady-state thermodynamics thus constructed on a Markov process is quite universal, no matter whether it is written in the form of C-K equation, master equation or F-P equation.

There are two important generalizations of the current study. We note the F-P equation is a special truncation of the master equation by just keeping the first two jump moments. Systematical explorations on high orders of truncation have been done and known as the Kramer equation in the literature.\cite{van1983} It remains to show a similar steady-state thermodynamic formalism could be defined and keeps in accordance with that of the master equation. The other non-trivial generalization is related to the quantum version of master equation, such as the Nakajima-Zwanzig equation in the exact form, or Redfield equation and Lindbald equation in approximate forms. \cite{Thingna2013Reduced} In contrast to classical master equations, which are restricted to only diagonal elements, quantum master equations deal with the entire density matrix, including off-diagonal elements, and thus have far more fruitful contents. Related works are going on.

\section*{acknowledgment}
This work was supported by the Tsinghua University Initiative Scientific Research Program (Grants 20151080424). L. H. would like to thank Prof. Wen-An Yong for his stimulating discussions.

\section*{appendix}

The limit of entropy production rate for the master equation is calculated in the main text, here we are going to prove remaining relations in Theorem \ref{thm2}, including the limit of entropy flow rate, adiabatic entropy production rate, non-adiabatic entropy production rate and excess entropy change rate.

\begin{proof}
Note that the entropy flow rate, adiabatic entropy production rate, non-adiabatic entropy production rate and excess entropy change rate for the master equation have a similar form,
that is, the integrand is a product of the thermodynamic flux and a thermodynamic force.
Then we have
\begin{eqnarray*}
\frac{dS^M_e}{dt}&=&\frac{1}{2} \iint dxdy J^{M}(x,y,t)A_2^{M}(x,y,t), \quad \frac{dS^M_{ad}}{dt}=\frac{1}{2} \iint dxdy J^{M}(x,y,t)A_3^{M}(x,y,t) , \\
\frac{dS^M_{na}}{dt}&=&\frac{1}{2} \iint dxdy J^{M}(x,y,t)A_4^{M}(x,y,t), \quad \frac{dS^M_{ex}}{dt}=\frac{1}{2} \iint dxdy J^{M}(x,y,t)A_5^{M}(x,y,t) ,
\end{eqnarray*}
where $J^{M}(x,y,t)=W_{\Omega}(x|y)P_{\Omega}(y,t) - W_{\Omega}(y|x)P_{\Omega}(x,t)$ denotes the thermodynamic flux and $A_2^{M}(x,y,t)=\ln \frac{W_{\Omega}(y|x)}{W_{\Omega}(x|y)}$, $A_3^{M}(x,y,t)=\ln \frac{W_{\Omega}(x|y)P_{\Omega}^s(y)}{W_{\Omega}(y|x)P_{\Omega}^s(x)}$, $A_4^{M}(x,y,t)=\ln \frac{P_{\Omega}^s(x)P_{\Omega}(y,t)}{P_{\Omega}^s(y)P_{\Omega}(x,t)}$, $A_5^{M}(x,y,t)=\ln \frac{P_{\Omega}^s(y)}{P_{\Omega}^s(x)}$ denote thermodynamic forces.

Expand $A_i^{M}(x,y,t)~(i=2,3,4,5)$ into Taylor series with respect to $X$ as
\begin{align*}
&{A_2^{M}}(x,y,t)\\
&=  \ln \big[ \Phi_0(X,-R) + {\Omega}^{-1}\Phi_1(X,-R) + O({\Omega}^{-2}) \big]
-\ln \big[\Phi_0(X-R,R) + {\Omega}^{-1}\Phi_1(X-R,R) + O({\Omega}^{-2}) \big] \\
&= \ln \big[ \Phi_0(X,-R) + {\Omega}^{-1}\Phi_1(X,-R) + O({\Omega}^{-2}) \big]
-\ln \big[\Phi_0(X,R) + {\Omega}^{-1}\Phi_1(X,R)-{R}{\partial_X} \Phi_0(X,R)+ O({\Omega}^{-2}) \big] \\
&=\frac{[\Phi_0(X,-R)-\Phi_0(X,R)] +{R}{\partial_X} \Phi_0(X,R) + {\Omega}^{-1} [\Phi_1(X,-R)-\Phi_1(X,R)]}{\Phi_0(X,-R) + {\Omega}^{-1}\Phi_1(X,-R) }+ O({\Omega}^{-2}) \\
&\equiv \frac{I_{21}(X,R)}{I_{22}(X,R)}+O({\Omega}^{-2}),\\
&{A_3^{M}}(x,y,t)\\
&= \frac{[\Phi_0(X,R)-\Phi_0(X,-R)](p_0^s+{\Omega}^{-1}p_1^s) -{R}{\partial_X} [\Phi_0(X,R)p_0^s]
+{\Omega}^{-1} [\Phi_1(X,R)-\Phi_1(X,-R)]p_0^s}{\Phi_0(X,-R)p_0^s + {\Omega}^{-1}\Phi_1(X,-R)p_0^s+ {\Omega}^{-1}\Phi_0(X,-R)p_1^s} +O({\Omega}^{-2}) \\
&\equiv \frac{I_{31}(X,R)}{I_{32}(X,R)}+O({\Omega}^{-2}),\\
&{A_4^{M}}(x,y,t)\\
&=-\ln \frac{p_0^s(X-R)+ {\Omega}^{-1}p_1^s(X-R)+O({\Omega}^{-2})}{p_0^s(X)+ {\Omega}^{-1}p_1^s(X)+O({\Omega}^{-2})}
  +\ln \frac{p_0(X-R)+ {\Omega}^{-1}p_1(X-R)+O({\Omega}^{-2})}{p_0(X)+ {\Omega}^{-1}p_1(X)+O({\Omega}^{-2})} \\
&= R \partial_X (\ln p_0^s) -  R \partial_X (\ln p_0) +O({\Omega}^{-2})
\equiv {I_{41}(X,R)}+O({\Omega}^{-2}),\\
&{A_5^{M}}(x,y,t)\\
&= \ln \frac{p_0^s(X-R)+ {\Omega}^{-1}p_1^s(X-R)+O({\Omega}^{-2})}{p_0^s(X)+ {\Omega}^{-1}p_1^s(X)+O({\Omega}^{-2})}
= -R \partial_X (\ln p_0^s) + O({\Omega}^{-2})
\equiv {I_{51}(X,R)}+O({\Omega}^{-2}),
\end{align*}
while the thermodynamic flux $J^{M}(x,y,t)$ becomes
\begin{align*}
{J^{M}}(x,y,t)
&=[\Phi_0(X,R)-\Phi_0(X,-R)]p_0  -{R}{\partial_X} [\Phi_0(X,R)p_0]
+{\Omega}^{-1} [\Phi_1(X,R)-\Phi_1(X,-R)]p_0   \\
& +{\Omega}^{-1}[\Phi_0(X,R)-\Phi_0(X,-R)]p_1 + O({\Omega}^{-2})\equiv I_{11}(X,R)+ O({\Omega}^{-2}).
\end{align*}

In the limit of $\Omega\rightarrow\infty$, the transition rates $\Phi_i(X,R)$ approach to delta functions with respect to the jump length $R$ (since $R=r/\Omega\rightarrow0$). Then
\begin{align*}
\lim_{\Omega \rightarrow \infty} \bigg({\Omega}^{-1}   \frac{dS^M_e}{dT}\bigg)
&=\frac{1}{2} \iint dXdR\frac{{\Omega}^2 R^2 I_{11}(X,R)I_{21}(X,R)}{I_{22}(X,R)R^2}\\
&=\frac{1}{2} \int dX~\frac{ \big[{\Omega}  \int_{V_0} dRI_{11}(X,R)R \big] \big[{\Omega}  \int_{V_0} dRI_{21}(X,R)R \big]}{ \int_{V_0} dR I_{22}(X,R)R^2 },\\
\lim_{\Omega \rightarrow \infty} \bigg({\Omega}^{-1}   \frac{dS^M_{ad}}{dT}\bigg)
&=\frac{1}{2} \iint dXdR\frac{{\Omega}^2 R^2 I_{11}(X,R)I_{31}(X,R)}{I_{32}(X,R)R^2}\\
&=\frac{1}{2} \int dX~\frac{ \big[{\Omega}  \int_{V_0} dRI_{11}(X,R)R \big] \big[{\Omega}  \int_{V_0} dRI_{31}(X,R)R \big]}{ \int_{V_0} dR I_{32}(X,R)R^2 }.
\end{align*}
Since
\begin{align*}
&\Omega \int_{V_0} dR [I_{11}(X,R)R]
= 2{\Omega}^{-1} \alpha_{1,1}(X)p_0 -2{\Omega}^{-1} \partial_X [\alpha_{2,0}(X)p_0] = 2{\Omega}^{-1} J^F(X,T),\\
&\Omega \int_{V_0} dR [I_{21}(X,R)R]\\
&=\Omega \int_{V_0} dR  \big\{ [\Phi_0(X,-R)-\Phi_0(X,R)]R +{R^2}{\partial_X} \Phi_0(X,R) + {\Omega}^{-1} [\Phi_1(X,-R)-\Phi_1(X,R)]R \big\}\\
&= 2{\Omega}^{-1} \partial_X [\alpha_{2,0}(X)]-2{\Omega}^{-1}\alpha_{1,1}(X),\\
&\int_{V_0} dR[I_{22}(X,R)R^2]
=\int_{V_0}  dR\big[\Phi_0(X,-R)R^2 + {\Omega}^{-1}\Phi_1(X,-R)R^2 \big]
= 2{\Omega}^{-2}\alpha_{2,0}(X)+O({\Omega}^{-3}), \\
&\Omega \int_{V_0} dR [ I_{31}(X,R)R]
= 2{\Omega}^{-1}\alpha_{1,1}(X)p_0^s-2{\Omega}^{-1} \partial_X [\alpha_{2,0}(X)p_0^s],\\
&\int_{V_0} dR[I_{32}(X,R)R^2]
= 2{\Omega}^{-2}\alpha_{2,0}(X)p_0^s +O({\Omega}^{-3}),
\end{align*}
the volume densities of the entropy flow rate and adiabatic entropy production rate for the master equation become
\begin{align*}
\lim_{\Omega \rightarrow \infty} \bigg({\Omega}^{-1} \frac{dS^M_e }{dT}\bigg)
&=\int dX J^F\cdot \frac{{\partial_X} \alpha_{2,0}(X) - \alpha_{1,1}(X)}{\alpha_{2,0}(X)}\equiv \frac{dS^F_e}{dT},\\
\lim_{\Omega \rightarrow \infty} \bigg({\Omega}^{-1} \frac{dS^M_{ad}}{dT}\bigg)
&=\int dX J^F\cdot \frac{\alpha_{1,1}(X)p_0^s- \partial_X [\alpha_{2,0}(X)p_0^s]}{\alpha_{2,0}(X) p_0^s}\equiv \frac{dS^F_{ad}}{dT},
\end{align*}
which emerge as the entropy flow rate and adiabatic entropy production rate for the corresponding F-P equation respectively.
Further note that ${I_{41}(X,R)}/{R}= {\partial_X} (\ln p_0^s)- {\partial_X} (\ln p_0)$ and ${I_{51}(X,R)}/{R}=-\partial_X (\ln p_0^s)$ are independent of $R$, we have
\begin{align*}
\lim_{\Omega \rightarrow \infty} \bigg({\Omega}^{-1}   \frac{dS^M_{na}}{dT}\bigg)
&=\frac{1}{2} \int dX~  \big[{\Omega} ~\frac{I_{41}(X,R)}{R} \big] \cdot \big[{\Omega}  \int_{V_0} dRI_{11}(X,R)R \big]\\
&=\int dX J^F\cdot [{\partial_X} (\ln p_0^s)- {\partial_X} (\ln p_0)] \equiv \frac{dS^F_{na}}{dT} ,\\
\lim_{\Omega \rightarrow \infty} \bigg({\Omega}^{-1}   \frac{dS^M_{ex}}{dT}\bigg)
&=\frac{1}{2} \int dX~  \big[{\Omega} ~\frac{I_{51}(X,R)}{R} \big] \cdot \big[{\Omega}  \int_{V_0} dRI_{11}(X,R)R \big]\\
&=\int dX J^F\cdot [-\partial_X (\ln p_0^s)] \equiv \frac{dS^F_{ex}}{dT}.
\end{align*}
This completes the proof.
\end{proof}

\bibliography{fp}
\bibliographystyle{aipnum4-1}
\end{document}